\documentclass[11pt,a4paper,final]{article}
\usepackage{amsmath}%
\usepackage{amstext}%
\usepackage{amssymb}%
\usepackage{showkeys}%
\usepackage{epsfig}%
\usepackage{graphicx}%
\usepackage{multicol}
\usepackage{xcolor}
\usepackage{apacite}
\usepackage{amsthm}
\usepackage{actuarialangle}
\usepackage{url}
\usepackage[top=1in, bottom=1.25in, left=0.8in, right=0.8in]{geometry}

\theoremstyle{prop}
\theoremstyle{proof}

\newtheorem{lem}{Lemma}

\begin{document}

\title{The Merton's Default Risk Model for Public Company}
\author{Battulga Gankhuu\footnote{
Department of Applied Mathematics, National University of Mongolia, Ulaanbaatar, Mongolia;
E-mail: battulgag@num.edu.mn}}
\date{}


\maketitle

\begin{abstract}
In this paper, we developed the Merton's structural model for public companies under an assumption that liabilities of the companies are observed. Using \citeauthor{Campbell88}'s approximation method, we obtain formulas of risk--neutral equity and liability values and default probabilities for the public companies. Also, the paper provides ML estimators of suggested model's parameters.\\[3ex]


\end{abstract}

\section{Introduction}

Dividend discount models (DDMs), first introduced by \citeA{Williams38} are common methods for equity valuation. The basic idea is that the market value of an equity of a firm is equal to the present value of a sum of dividend paid by the firm and market value of the firm, which correspond to the next period. If the firm is financed by liabilities, which are publicly traded in the exchanges, the same idea can be used to value the liabilities, see \citeA{Battulga23b}. As the outcome of DDMs depends crucially on dividend payment forecasts, most research in the last few decades has been around the proper estimations of dividend development. Also, parameter estimation of DDMs is a challenging task. Recently, \citeA{Battulga22a} introduced parameter estimation methods for practically popular DDMs. To estimate parameters of the required rate of return, \citeA{Battulga23b} used the maximum likelihood method and Kalman filtering. Reviews of some existing DDMs that include deterministic and stochastic models can be found in \citeA{dAmico20a} and \citeA{Battulga22a}. 

Existing stochastic DDMs have one common disadvantage: If dividend and debt payments have chances to take negative values, then the market values of the firm's equity and liabilities can take negative values with a positive probability, which is the undesirable property for the market values. A log version of the stochastic DDM, which is called by dynamic Gordon growth model was introduced by \citeA{Campbell88}, who derived a connection between log price, log dividend, and log return by approximation. Since their model is in a log framework, the stock prices and dividends get positive values. For this reason, by augmenting the dynamic Gordon growth model, \citeA{Battulga24d} obtained pricing and hedging formulas of European options and equity--linked life insurance products for public companies. For private companies, using the log private company valuation model, based on the dynamic Gordon growth model, \citeA{Battulga24c} developed closed--form pricing and hedging formulas for the European options and equity--linked life insurance products and valuation formula.

Sudden and dramatic changes in the financial market and economy are caused by events such as wars, market panics, or significant changes in government policies. To model those events, some authors used regime--switching models. The regime--switching model was introduced by seminal works of \citeA{Hamilton89,Hamilton90} (see also books of \citeA{Hamilton94} and \citeA{Krolzig97}) and the model is hidden Markov model with dependencies, see \citeA{Zucchini16}. However, Markov regime--switching models have been introduced before Hamilton (1989), see, for example,  \citeA{Goldfeld73}, \citeA{Quandt58}, and \citeA{Tong83}. The regime--switching model assumes that a discrete unobservable Markov process generates switches among a finite set of regimes randomly and that each regime is defined by a particular parameter set. The model is a good fit for some financial data and has become popular in financial modeling including equity options, bond prices, and others. Recently, under normal framework, \citeA{Battulga22b} obtained pricing and hedging formulas for the European options and equity--linked life insurance products by introducing a DDM with regime--switching process. Also, \citeA{Battulga24a} developed option pricing formulas for some frequently used options by using Markov--Switching Vector Autoregressive process. To model required rate of return on stock, \citeA{Battulga23b} applied a two--regime model. The result of the paper reveals that the regime--switching model is good fit for the required rate of return.

Default risk is a possibility that a borrower fails to make full and timely payments of principal and interest, which are stated in the debt contract. The structural model of default risk relates to option pricing. In this model, a default threshold, which represents the liabilities of the company is seen as a strike price and a asset value of the company is seen as underlying asset of the European option. For this reason, this approach is also referred to as the firm--value approach or the option--theoretic approach. Original idea of the structural model goes back to \citeA{Black73} and \citeA{Merton74}. \citeA{Black73} developed a closed--form formula for evaluating the European option and \citeA{Merton74} obtained pricing formula for the liabilities of a company under Black--Scholes framework. \citeA{Battulga22b} tried to estimate default probability using regime--switching process.

This paper is organized as follows. In section 2 of the paper, we develop stochastic DDM for market values of equities and liabilities of companies using the \citeauthor{Campbell88}'s \citeyear{Campbell88} approximation method. Then, we model the market values of assets of the companies using the approximation method once again. In section 3, we obtain pricing formulas of the European call and put options on the market value of the asset. After that, we develop formulas of risk--neutral equity and debt values, and default probability. In section 4, we study ML estimators of suggested model's parameters. In section 5, we conclude the study. Finally, in section 6, we provide Lemmas, which is used in the paper.

\section{Market Value Model of Equity and Liability}

Let $(\Omega,\mathcal{H}_T,\mathbb{P})$ be a complete probability space, where $\mathbb{P}$ is a given physical or real--world probability measure and $\mathcal{H}_T$ will be defined below. To introduce a regime--switching in Merton's default risk model, we assume that $\{s_t\}_{t=1}^T$ is a homogeneous Markov chain with $N$ state and $\mathsf{P}:=\{p_{ij}\}_{i=0,j=1}^N$ is a random transition probability matrix, where $p_0:=(p_{01},\dots,p_{0N})$ is a $(1\times N)$ initial probability vector. In this paper, we assume that market values of equities and liabilities of companies are observed. For a case that market values of equities and liabilities are both unobserved, we refer to \citeA{Battulga24e}.

Dividend discount models (DDMs), first introduced by \citeA{Williams38}, are a popular tool for equity valuation. The basic idea of all DDMs is that the market value of equity at time $t-1$ of the firm equals the sum of the market value of equity at time $t$ and dividend payment at time $t$ discounted at risk--adjusted rate (required rate of return on stock). Let us assume there are $n$ companies. Therefore, for successive market values of equity of $i$--th company, the following relation holds 
\begin{equation}\label{05001}
V_{i,t}^e=(1+k_{i,t}^e)V_{i,t-1}^e-p_{i,t}^e,~~~t=1,\dots,T,
\end{equation}
where $k_{i,t}^e$ is the required rate of return on the equity (investors) at regime $s_t$, $V_{i,t}^e$ is the market value of equity, and $p_{i,t}^e$ is the dividend payment for investors, respectively, at time $t$ of $i$--th company. On the other hand, to model market values of liabilities of the company, it is the well known fact that successive values of a debt of company or individual is given by the following equation
\begin{equation}\label{05002}
D_t=(1+i)D_{t-1}-d_t
\end{equation}
where $D_t$ is a debt value at time $t$, $d_t$ is a debt payment at time $t$, and $i$ is a interest rate of the debt, see, e.g., \citeA{Gerber97}. Note that $D_t$ represents the principal outstanding, that is, the remaining debt immediately after $d_t$ has been paid and debt equation \eqref{05002} shares same formula with market value of equity given in equation \eqref{05001}. The idea of equation \eqref{05002} can be used to model a value of liabilities of the company, namely,
\begin{equation}\label{05003}
V_{i,t}^\ell=(1+k_{i,t}^\ell)V_{i,t-1}^\ell-p_{i,t}^\ell, ~~~t=1,\dots,T,
\end{equation}
where $k_{i,t}^\ell$ is a required rate of return on the liability (debtholders) at regime $s_t$, $V_{i,t}^\ell$ is a market value of the liability, and $p_{i,t}^\ell$ is a debt payment, which includes interest payment for debtholders, respectively, at time $t$ of the company, see \citeA{Battulga23b}.

To keep notations simple, let $V_t^e:=(V_{1,t}^e,\dots,V_{n,t}^e)'$ be an $(n\times 1)$ vector of market values of equities, $V_t^\ell:=(V_{1,t}^\ell,\dots,V_{n,t}^\ell)'$ be an $(n\times 1)$ vector of market values of liabilities, $k_t^e:=(k_{1,t}^e,\dots,k_{n,t}^e)'$ be an $(n\times 1)$ vector of required rate of returns on equities, $k_t^\ell:=(k_{1,t}^\ell,\dots,k_{n,t}^\ell)'$ be an $(n\times 1)$ vector of required rate of returns on liabilities, $p_t^e:=(p_{1,t}^e,\dots,p_{n,t}^e)'$ be an $(n\times 1)$ vector of the dividend payments, and $p_t^\ell:=(p_{1,t}^\ell,\dots,p_{n,t}^\ell)'$ be an $(n\times 1)$ vector of the debt payments, respectively, at time $t$, $I_n$ be an $(n\times n)$ identity matrix, $i_n:=(1,\dots,1)'$ be an $(n\times 1)$ vector, whose all elements equal one.

If payments of dividend and debt have chances to take negative values, then the market values of equity and liability of a company can take negative values with a positive probability, which is the undesirable property for the market values of the equity and liability. That is why, we follow the idea in \citeA{Campbell88}. As a result, the market values of equity and liability of the company take positive values. Following the idea in \citeA{Campbell88}, one can obtain the following approximation
\begin{equation}\label{05004}
\exp\{\tilde{k}_t\}=(V_t+p_t)\oslash V_{t-1}\approx \exp\Big\{\tilde{V}_t-\tilde{V}_{t-1}+\ln(g_t)+G_t^{-1}(G_t-I_{2n})\big(\tilde{p}_t-\tilde{V}_t-\mu_t\big)\Big\},
\end{equation}
where $\oslash$ is a component--wise division of two vectors, $\tilde{k}_t:=\big((\ln(i_n+k_t^e))',(\ln(i_n+k_t^\ell))'\big)'$ is a ($2n\times 1$) log required rate of return process at time $t$, $V_t:=\big((V_t^e)',(V_t^\ell)'\big)'$ is a $(2n\times 1)$ market value process at time $t$, $p_t:=\big((p_t^e)',(p_t^\ell)'\big)'$ is a $(2n\times 1)$ payment process at time $t$, $\tilde{V}_t:=\ln(V_t)$ is a $(2n\times 1)$ log market value process at time $t$, $\tilde{p}_t:=\ln(p_t)$ is a $(2n\times 1)$ log payment process at time $t$, $\mu_t:=\mathbb{E}\big[\tilde{p}_t-\tilde{V}_t\big|\mathcal{F}_0\big]$ is a $(2n\times 1)$ mean log payment--to--market value process at time $t$ of the companies and $\mathcal{F}_0$ is initial information, which will be defined below, $g_t:=i_{2n}+\exp\{\mu_t\}$ is a $(2n\times 1)$ linearization parameter, and $G_t:=\text{diag}\{g_t\}$ is a $(2n\times 2n)$ diagonal matrix, whose diagonal elements are $g_t$. As a result, for the log market value process at time $t$, the following approximation holds
\begin{equation}\label{05005}
\tilde{V}_t\approx G_t(\tilde{V}_{t-1}-\tilde{p}_t+\tilde{k}_t)+\tilde{p}_t-h_t.
\end{equation}
where $h_t:=G_t\big(\ln(g_t)-\mu_t\big)+\mu_t$ is a linearization parameter and the model is called by dynamic Gordon growth model, see \citeA{Campbell88}. We assume that values of the log payment process $\tilde{p}_t$ are known at time 0. For a quality of the approximation, we refer to \citeA{Campbell97}. To estimate parameters of the dynamic Gordon growth model, to price the Black--Scholes call and put options on asset values of the companies, and to calculate a joint default probability of the companies, we suppose that the log required rate of return process at time $t$ is represented by a sum of deterministic process, a term, which depends on log spot interest rates, and white noise process, namely, 
\begin{equation}\label{05006}
\tilde{k}_t=C_{k,s_t}\psi_t+\delta \tilde{r}_t+u_t,
\end{equation}
where $\psi_t:=(\psi_{1,t},\dots,\psi_{l,t})'$ is an $(l\times 1)$ vector, which consists of exogenous variables, $C_{k,s_t}$ is an $(n\times l)$ random matrix at regime $s_t$, $\delta:=(0,i_n')'$ is a $(2n\times 1)$ vector, whose first $n$ elements are zero and others are one, $\tilde{r}_t:=\ln(1+r_t)$ is a log spot interest rate at time $t$, $r_t$ is a spot interest rate for borrowing and lending over a period $(t,t+1]$, $u_t$ is an $(n\times 1)$ white noise process. In this case, equation \eqref{05005} becomes 
\begin{equation}\label{05007}
\tilde{V}_t=G_t(\tilde{V}_{t-1}-\tilde{p}_t+C_{k,s_t}\psi_t+\delta \tilde{r}_t)+\tilde{p}_t-h_t+G_tu_t.
\end{equation}

Now, we model spot interest rate $r_t$. By using the Dickey--Fuller test, it can be confirm that quarterly log spot interest rate is the unit--root process with drift, see data IRX of Yahoo Finance. Also, due to the fact that log return of quarterly S\&P 500 index is stationary process, see data SPX of Yahoo Finance, we model the log required rate of return on equity by a trend stationary process. Moreover, there may exist a cointegration between the log required rate of return on debtholders and the spot interest rate. For this reason, we model the required rate of return process by equation \eqref{05006}. Consequently, we model the log spot rate by the following equation
\begin{equation}\label{05008}
\tilde{r}_t=c_{r,s_t}'\psi_t+\tilde{r}_{t-1}+v_t,
\end{equation}
where $c_{r,s_t}$ is an $(l\times 1)$ random vector at regime $s_t$ and $v_t$ is a white noise process.

As a result, by combining equations \eqref{05007} and \eqref{05008}, we arrive the following system
\begin{equation}\label{05009}
\begin{cases}
\tilde{V}_t=\nu_{V,t}+G_t\tilde{V}_{t-1}+G_t\delta \tilde{r}_t+G_tu_t\\
\tilde{r}_t=\nu_{r,t}+\tilde{r}_{t-1}+v_t,
\end{cases}~~~\text{for}~t=1,\dots,T
\end{equation}
under the real probability measure $\mathbb{P}$, where $\nu_{V,t}:=G_tC_{k,s_t}\psi_t-(G_t-I_{2n})\tilde{p}_t-h_t$ is an $(2n\times 1)$ intercept process of the log market value process $\tilde{V}_t$ and $\nu_{r,t}:=c_{r,s_t}'\psi_t$ is an $(1\times 1)$ intercept process of the log spot rate process $r_t$. Let us denote a dimension of system \eqref{05011} by $\tilde{n}$, that is, $\tilde{n}:=2n+1$.

Finally, we model the market values of the assets of the companies. Since the market values of the assets equal sums of the market values of equities and liabilities of the companies, we have 
$$V_t^a=V_t^e+V_t^\ell,$$
where $V_t^a$ is an $(n\times 1)$ asset value process at time $t$ of the companies. Using the same approximation method, a log asset value process of the companies is approximated by the following equation
\begin{eqnarray}\label{05010}
\tilde{V}_t^a:=\ln(V_t^e+V_t^\ell)&\approx &(G_t^a)^{-1}\tilde{V}_t^e+\big(I_n-(G_t^a)^{-1}\big)\tilde{V}_t^\ell+(G_t^a)^{-1}h_t^a\nonumber\\
&=&W_t^a\tilde{V}_t+(G_t^a)^{-1}h_t^a 
\end{eqnarray}
where $\mu_t^a:=\mathbb{E}[\tilde{V}_t^\ell-\tilde{V}_t^e|\mathcal{F}_0]$ is an $(n\times 1)$ mean log liability value--to--equity value ratio, $g_t^a:=i_n+\exp\{\mu_t^a\}$ and $h_t^a:=G_t^a(\ln(g_t^a)-\mu_t^a)+\mu_t^a$ are $(n\times 1)$ linearization parameters for the log asset process, $G_t^a:=\text{diag}\{g_t^a\}$ is an $(n\times n)$ diagonal matrix, and $W_t^a:=\big[(G_t^a)^{-1}:I_n-(G_t^a)^{-1}\big]$ is an $(n\times 2n)$ weight matrix of the approximation, respectively, at time $t$ of the company.

The stochastic properties of system \eqref{05009} is governed by the random vectors $\{u_1,\dots,u_T, v_1,\dots,v_T\}.$ We assume that for $t=1,\dots,T$, conditional on information $\mathcal{H}_0$, which is defined below, the white noise process $\xi_t:=(u_t',v_t)'$ are mutually independent and follows normal distribution, namely,
\begin{equation}\label{05011}
\xi_t~|~\mathcal{H}_0\sim \mathcal{N}(0,\Sigma_{s_t})
\end{equation}
under the real probability measure $\mathbb{P}$, where
\begin{equation}\label{05012}
\Sigma_{s_t}=\begin{bmatrix}
\Sigma_{uu,s_t} & \Sigma_{uv,s_t}\\
\Sigma_{vu,s_t} & \Sigma_{vv,s_t}
\end{bmatrix}.
\end{equation}
is a covariance matrix of the $(\tilde{n}\times 1)$ white noise process $\xi_t$.

\section{Merton's Structural Model}

The Merton's model \citeyear{Merton74} is one of the structural models used to measure credit risk. \citeA{Merton74} was aim to value the liabilities of a specific company. As mentioned above the model connects the European call and put options.  The European call and put options are contracts that give their owner the right, but not the obligation, to buy or sell shares of a stock of a company at a predetermined price by a specified date. Let us start this section by considering a valuation method of the European options on the asset values of the companies.

Let $T$ be a time to maturity of the European call and put options, $x_t:=(\tilde{V}_t',\tilde{r}_t)'$ be $(\tilde{n}\times 1)$ process at time $t$ of endogenous variables, and $C_{s_t}:=\big[C_{k,s_t}':c_{r,s_t}\big]'$ be random coefficient matrix at regime $s_t$. We introduce stacked vectors and matrices: $x:=(x_1',\dots,x_T')'$, $s:=(s_1,\dots,s_T)'$, $C_s:=[C_{s_1}:\dots:C_{s_T}]$, and $\Sigma_s:=[\Sigma_{s_1}:\dots:\Sigma_{s_T}]$. We suppose that the white noise process $\{\xi_t\}_{t=1}^T$ is independent of the random coefficient matrix $C_s$, random covariance matrix $\Sigma_s$, random transition matrix $\mathsf{P}$, and regime--switching vector $s$ conditional on initial information $\mathcal{F}_0:=\sigma(x_0,\psi_{1},\dots,\psi_T,\tilde{p}_1,\dots,\tilde{p}_T)$. Here for a generic random vector $X$, $\sigma(X)$ denotes a $\sigma$--field generated by the random vector $X$, $\psi_1,\dots,\psi_T$ are values of exogenous variables and they are known at time zero, and according to the assumption, values of $\tilde{p}_1,\dots,\tilde{p}_T$ are known at time zero. We further suppose that the transition probability matrix $\mathsf{P}$ is independent of the random coefficient matrix $C_s$ and covariance matrix $\Sigma_s$ given initial information $\mathcal{F}_0$ and regime--switching process $s_t$.

To ease of notations, for a generic vector $o=(o_1',\dots,o_T')'$, we denote its first $t$ and last $T-t$ sub vectors by $\bar{o}_t$ and $\bar{o}_t^c$, respectively, that is, $\bar{o}_t:=(o_1',\dots,o_t')'$ and $\bar{o}_t^c:=(o_{t+1}',\dots,o_T')'$. We define $\sigma$--fields: for $t=0,\dots,T$, $\mathcal{F}_{t}:=\mathcal{F}_0\vee\sigma(\bar{x}_{t})$ and $\mathcal{H}_{t}:=\mathcal{F}_t\vee \sigma(C_s)\vee \sigma(\Sigma_s)\vee \sigma(\mathsf{P})\vee\sigma(s)$, where for generic sigma fields $\mathcal{O}_1,\dots,\mathcal{O}_k$, $\vee_{i=1}^k \mathcal{O}_i $ is the minimal $\sigma$--field containing the $\sigma$--fields $\mathcal{O}_i$, $i=1,\dots,k$. Observe that $\mathcal{F}_{t}\subset \mathcal{H}_{t}$ for $t=0,\dots,T$. 

For the first--order Markov chain, a conditional probability that the regime at time $t+1$, $s_{t+1}$ equals some particular value conditional on the past regimes, $\bar{s}_t$, transition probability matrix $\mathsf{P}$, and initial information $\mathcal{F}_0$ depends only through the most recent regime at time $t$, $s_t$, transition probability matrix $\mathsf{P}$, and initial information $\mathcal{F}_0$ that is,
\begin{equation}\label{06011}
p_{s_ts_{t+1}}:=\mathbb{P}[s_{t+1}=s_{t+1}|s_t=s_t,\mathsf{P},\mathcal{F}_0]=\mathbb{P}\big[s_{t+1}=s_{t+1}|\bar{s}_t=\bar{s}_t,\mathsf{P},\mathcal{F}_0\big]
\end{equation} 
for $t=0,\dots,T-1$, where $p_{0s_1}=p_{s_0s_1}=\mathbb{P}[s_1=s_1|\mathsf{P},\mathcal{F}_0]$ is the initial probability. A distribution of a white noise vector $\xi:=(\xi_1',\dots,\xi_T')'$ is given by
\begin{equation}\label{06012}
\xi=(\xi_1',\dots,\xi_T')'~|~\mathcal{H}_0\sim \mathcal{N}(0,\bar{\Sigma}),
\end{equation}
where $\bar{\Sigma}:=\text{diag}\{\Sigma_{s_1},\dots,\Sigma_{s_T}\}$ is a block diagonal matrix. 

To remove duplicates in the random coefficient matrix $(C_s,\Sigma_s)$, for a generic regime--switching vector with length $k$, $o=(o_1,\dots,o_k)'$, we define sets
\begin{equation}\label{08006}
\mathcal{A}_{\bar{o}_t}:=\mathcal{A}_{\bar{o}_{t-1}}\cup\big\{o_t\in \{o_1,\dots,o_k\}\big|o_t\not \in \mathcal{A}_{\bar{o}_{t-1}}\big\},~~~t=1,\dots,k,
\end{equation}
where for $t=1,\dots,k$, $o_t\in \{1,\dots,N\}$ and an initial set is the empty set, i.e., $\mathcal{A}_{\bar{o}_0}=\O$. The final set $\mathcal{A}_o=\mathcal{A}_{\bar{o}_k}$ consists of different regimes in regime vector $o=\bar{o}_k$ and $|\mathcal{A}_o|$ represents a number of different regimes in the regime vector $o$. Let us assume that elements of sets $\mathcal{A}_s$, $\mathcal{A}_{\bar{s}_t}$, and difference sets between the sets $\mathcal{A}_{\bar{s}_t^c}$ and $\mathcal{A}_{\bar{s}_t}$ are given by $\mathcal{A}_s=\{\hat{s}_1,\dots,\hat{s}_{r_{\hat{s}}}\}$, $\mathcal{A}_{\bar{s}_t}=\{\alpha_1,\dots,\alpha_{r_\alpha}\}$, and $\mathcal{A}_{\bar{s}_t^c}\backslash \mathcal{A}_{\bar{s}_t}=\{\delta_1,\dots,\delta_{r_\delta}\}$, respectively, where $r_{\hat{s}}:=|\mathcal{A}_s|$, $r_\alpha:=|\mathcal{A}_{\bar{s}_t}|$, and $r_\delta:=|\mathcal{A}_{\bar{s}_t^c}\backslash \mathcal{A}_{\bar{s}_t}|$ are numbers of elements of the sets, respectively. We introduce the following regime vectors: $\hat{s}:=(\hat{s}_1,\dots,\hat{s}_{r_{\hat{s}}})'$ is an $(r_{\hat{s}}\times 1)$ vector, $\alpha:=(\alpha_1,\dots,\alpha_{r_\alpha})'$ is an $(r_\alpha\times 1)$ vector, and $\delta=(\delta_1,\dots,\delta_{r_\delta})'$ is an $(r_\delta\times 1)$ vector. For the regime vector $a=(a_1,\dots,a_{r_a})' \in\{\hat{s},\alpha,\delta\}$, we also introduce duplication removed random coefficient matrices, whose block matrices are different:  $C_a=[C_{a_1}:\dots:C_{a_{r_a}}]$ is an $(\tilde{n}\times [lr_a])$ matrix, $\Sigma_a=[\Sigma_{a_1}:\dots:\Sigma_{a_{r_a}}]$ is an $(\tilde{n}\times [\tilde{n}r_a])$ matrix, and $(C_a,\Sigma_a)$. 

We assume that for given duplication removed regime vector $\hat{s}$ and initial information $\mathcal{F}_0$, the coefficient matrices $(C_{\hat{s}_1},\Sigma_{\hat{s}_1}),\dots,(C_{\hat{s}_{r_{\hat{s}}}},\Sigma_{\hat{s}_{r_{\hat{s}}}})$ are independent under the real probability measure $\mathbb{P}$. Under the assumption, conditional on $\hat{s}$ and $\mathcal{F}_0$, a joint density function of the random coefficient random matrix $(C_{\hat{s}},\Sigma_{\hat{s}})$ is represented by
\begin{equation}\label{08010}
f\big(C_{\hat{s}},\Sigma_{\hat{s}}\big|\hat{s},\mathcal{F}_0\big)=\prod_{t=1}^{r_{\hat{s}}}f\big(C_{\hat{s}_t},\Sigma_{\hat{s}_t}\big|\hat{s}_t,\mathcal{F}_0\big)
\end{equation}
under the real probability measure $\mathbb{P}$, where for a generic random vector $X$, we denote its density function by $f(X)$ under the real probability measure $\mathbb{P}$. Using the regime vectors $\alpha$ and $\delta$, the above joint density function can be written by 
\begin{equation}\label{08011}
f\big(C_{\hat{s}},\Sigma_{\hat{s}}\big|\hat{s},\mathcal{F}_0\big)=
f\big(C_{\alpha},\Sigma_{\alpha}\big|\alpha,\mathcal{F}_0\big)f_*\big(C_{\delta},\Sigma_{\delta}\big|\delta,\mathcal{F}_0\big)
\end{equation}
where the density function $f_*\big(C_{\delta},\Sigma_{\delta}\big|\delta,\mathcal{F}_0\big)$ equals
\begin{equation}\label{08012}
f_*\big(C_\delta,\Sigma_\delta\big|\delta,\mathcal{F}_0\big):=
\begin{cases}
f\big(C_\delta,\Sigma_\delta\big|\delta,\mathcal{F}_0\big),& \text{if}~~~r_\delta\neq 0,\\
1,& \text{if}~~~r_\delta= 0.
\end{cases}
\end{equation}

\subsection{Risk--Neutral Probability Measure}

To price the European call and put options, we need to change from the real probability measure to some risk--neutral measure. Let $D_t:=\exp\{-\tilde{r}_1-\dots-\tilde{r}_t\}=1\big/\prod_{s=1}^t(1+r_s)$ be a predictable discount process, where $\tilde{r}_t$ is the log spot interest rate at time $t$. According to \citeA{Pliska97} (see also \citeA{Bjork20}), for all companies, conditional expectations of the return processes $k_{i,t}^e=(V_{i,t}^e+p_{i,t}^e)/V_{i,t-1}^e-1$ and $k_{i,t}^\ell=(V_{i,t}^\ell+p_{i,t}^\ell)/V_{i,t-1}^\ell-1$ for $i=1,\dots,n$ must equal the spot interest rate $r_t$ under some risk--neutral probability measure $\tilde{\mathbb{P}}$ and a filtration $\{\mathcal{H}_t\}_{t=0}^T$. Thus, it must hold
\begin{equation}\label{05016}
\tilde{\mathbb{E}}\big[(V_t+p_t)\oslash V_{t-1}\big|\mathcal{H}_{t-1}\big]=\exp\big\{\tilde{r}_ti_{2n}\big\}
\end{equation}
for $t=1,\dots,T$, where $\tilde{\mathbb{E}}$ denotes an expectation under the risk--neutral probability measure $\mathbb{\tilde{P}}$. According to equation \eqref{05004}, condition \eqref{05016} is equivalent to the following condition
\begin{equation}\label{05017}
\tilde{\mathbb{E}}\big[\exp\big\{u_t-\big((i_{2n}-\delta)\tilde{r}_t-C_{k,s_t}\psi_t\big)\big\}\big|\mathcal{H}_{t-1}\big]=i_{2n}.
\end{equation}

It should be noted that condition \eqref{05017} corresponds only to the white noise random process $u_t$. Thus, we need to impose a condition on the white noise process $v_t$ under the risk--neutral probability measure. This condition is fulfilled by $\tilde{\mathbb{E}}[f(v_t)|\mathcal{H}_{t-1}]=\tilde{\theta}_t$ for any Borel function $f:\mathbb{R}\to\mathbb{R}$ and $\mathcal{H}_{t-1}$ measurable any random variable $\tilde{\theta}_t$. Because for any admissible choices of $\tilde{\theta}_t$, condition \eqref{05017} holds, the market is incomplete. But prices of the options are still consistent with the absence of arbitrage. For this reason, to price the options, in this paper, we will use a unique optimal Girsanov kernel process $\theta_t$, which minimizes the variance of a state price density process and relative entropy. According to \citeA{Battulga24a}, the optimal kernel process $\theta_t$ is obtained by
\begin{equation}\label{05018}
\theta_t=\Theta_t \bigg((i_{2n}-\delta)\tilde{r}_t-C_{k,s_t}\psi_t-\frac{1}{2}\mathcal{D}[\Sigma_{uu,s_t}]\bigg),
\end{equation}
where $\Theta_t=\big[G_t:(\Sigma_{vu,s_t}\Sigma_{uu,s_t}^{-1})'\big]'$ and for a generic square matrix $O$, $\mathcal{D}[O]$ denotes a vector, consisting of diagonal elements of the matrix $O$. Consequently, system \eqref{05009} can be written by
\begin{equation}\label{05019}
\begin{cases}
\tilde{V}_t=\tilde{\nu}_{V,t}+G_t\tilde{V}_{t-1}+G_ti_{2n} r_{t-1}+G_t\tilde{u}_t\\
F_t\tilde{r}_t=\tilde{\nu}_{r,t}+\tilde{r}_{t-1}+\tilde{v}_t,
\end{cases}~~~\text{for}~t=1,\dots,T
\end{equation}
under the risk--neutral probability measure $\tilde{\mathbb{P}}$, where $\tilde{\nu}_{V,t}:=-(G_t-I_n)\tilde{p}_t-\frac{1}{2}G_t\mathcal{D}[\Sigma_{uu,s_t}]-h_t$ is an $(2n\times 1)$ intercept process of the log market value process $\tilde{V}_t$, $F_t:=1-\Sigma_{vu,s_t}\Sigma_{uu,s_t}^{-1}(i_{2n}-\delta)$ is a $(1\times 1)$ process, and $\tilde{\nu}_{r,t}:=c_{r,s_t}'\psi_t-\Sigma_{vu,s_t}\Sigma_{uu,s_t}^{-1}\big(C_{k,s_t}\psi_t+\frac{1}{2}\mathcal{D}[\Sigma_{uu,s_t}]\big)$ is a $(1\times 1)$ intercept process of the log spot rate process $\tilde{r}_t$. It is worth mentioning that  a joint distribution of a random vector $\tilde{\xi}:=(\tilde{\xi}_1',\dots,\tilde{\xi}_T')'$ with $\tilde{\xi}_t:=(\tilde{u}_t',\tilde{v}_t)'$ equals the joint distribution of the random vector $\xi=(\xi_1',\dots,\xi_T')'$, that is,
\begin{equation}\label{05020}
\tilde{\xi}~|~\mathcal{H}_0\sim \mathcal{N}\big(0,\bar{\Sigma}\big)
\end{equation}
under the risk--neutral probability measure $\mathbb{\tilde{P}}$, see \citeA{Battulga24a}.

System \eqref{05019} can be written in VAR(1) form, namely
\begin{equation}\label{05021}
\tilde{Q}_{0,t}x_t=\tilde{\nu}_t+\tilde{Q}_{1,t}x_{t-1}+\mathsf{G}_t\tilde{\xi}_t
\end{equation} 
under the risk--neutral probability measure $\mathbb{\tilde{P}}$, where $\tilde{\nu}_t:=(\tilde{\nu}_{V,t}',\tilde{\nu}_{r,t})'$, and $\tilde{\xi}_t:=\big(\tilde{u}_t',\tilde{v}_t\big)'$ are intercept process and white noise processes of the VAR(1) process $x_t$, respectively, and
\begin{equation}\label{05022}
\tilde{Q}_{0,t}:=\begin{bmatrix}
I_{2n} & 0 \\
0 & F_t
\end{bmatrix},~~~\tilde{Q}_{1,t}:=\begin{bmatrix}
G_t & G_ti_{2n} \\
0 & 1
\end{bmatrix},~~~\text{and}~~~
\mathsf{G}_t=\begin{bmatrix}
G_t & 0\\
0 & 1
\end{bmatrix}
\end{equation}
are $(\tilde{n}\times \tilde{n})$ coefficient matrices. By repeating equation \eqref{05021}, one gets that for $i=t+1,\dots,T$,
\begin{equation}\label{05023}
x_i=\tilde{\Pi}_{t,i}x_t+\sum_{\beta=t+1}^i\tilde{\Pi}_{\beta,i}\tilde{\nu}_\beta+\sum_{\beta=t+1}^i\tilde{\Pi}_{\beta,i}\mathsf{G}_\beta\tilde{\xi}_\beta,
\end{equation}
where the coefficient matrices are for $\beta=t$,
\begin{equation}\label{05024}
\tilde{\Pi}_{\beta,i}:=\prod_{\alpha=\beta+1}^i\tilde{Q}_{0,\alpha}^{-1}\tilde{Q}_{1,\alpha}=\begin{bmatrix}
\displaystyle\prod_{\alpha=\beta+1}^iG_\alpha & \displaystyle \sum_{\alpha=\beta+1}^i\Bigg(\prod_{j_1=\alpha}^iG_{j_1}\Bigg)i_{2n}\Bigg(\prod_{j_2=\beta+1}^{\alpha-1} F_{j_2}^{-1}\Bigg)\\
\displaystyle 0 & \displaystyle\prod_{\alpha=\beta+1}^i F_\alpha^{-1}
\end{bmatrix}
\end{equation}
for $\beta=t+1,\dots,i-1$ ,
\begin{equation}\label{05024}
\tilde{\Pi}_{\beta,i}:=\Bigg(\prod_{\alpha=\beta+1}^i\tilde{Q}_{0,\alpha}^{-1}\tilde{Q}_{1,\alpha}\Bigg)\tilde{Q}_{0,\beta}^{-1}=\begin{bmatrix}
\displaystyle\prod_{\alpha=\beta+1}^iG_\alpha & \displaystyle \sum_{\alpha=\beta+1}^i\Bigg(\prod_{j_1=\alpha}^iG_{j_1}\Bigg)i_{2n}\Bigg(\prod_{j_2=\beta}^{\alpha-1} F_{j_2}^{-1}\Bigg)\\
\displaystyle 0 & \displaystyle\prod_{\alpha=\beta}^i F_\alpha^{-1}
\end{bmatrix},
\end{equation}
and for $\beta=i$,
\begin{equation}\label{ad001}
\tilde{\Pi}_{\beta,i}:=\tilde{Q}_{0,\beta}^{-1}=\begin{bmatrix}
I_{2n} & 0 \\
0 & F_{\beta}^{-1}
\end{bmatrix}.
\end{equation}
Here for a sequence of generic $(k\times k)$ square matrices $O_1,O_2,\dots$, the products mean that for $v\leq u$, $\prod_{j=v}^uO_j=O_u\dots O_v$ and for $v>u$, $\prod_{j=v}^uO_j=I_k$.

Therefore, conditional on the information $\mathcal{H}_t$, for $i=t+1,\dots,T$, a expectation at time $i$ and a conditional covariance matrix at times $i_1$ and $i_2$ of the process $x_t$ is given by the following equations
\begin{equation}\label{05025}
\tilde{\mu}_{i|t}(\mathcal{H}_t):=\mathbb{\tilde{E}}\big[x_i\big|\mathcal{H}_t\big]=\tilde{\Pi}_{t,i} x_t+\sum_{\beta=t+1}^i\tilde{\Pi}_{\beta,i}\tilde{\nu}_\beta
\end{equation}
and
\begin{equation}\label{05026}
\tilde{\Sigma}_{i_1,i_2|t}(\mathcal{H}_t):=\widetilde{\text{Cov}}\big[x_{i_1},x_{i_2}\big|\mathcal{H}_t\big]=\sum_{\beta=t+1}^{i_1\wedge i_2}\tilde{\Pi}_{\beta,i_1}\mathsf{G}_\beta\Sigma_{s_\beta}\mathsf{G}_\beta\tilde{\Pi}_{\beta,i_2}',
\end{equation}
where $i_1\wedge i_2$ is a minimum of $i_1$ and $i_2$. Consequently, due to equation \eqref{05023}, conditional on the information $\mathcal{H}_t$, a joint distribution of a random vector $\bar{x}_t^c:=(x_{t+1}',\dots,x_T')'$ is
\begin{equation}\label{05027}
\bar{x}_t^c~|~\mathcal{H}_t\sim \mathcal{N}\big(\tilde{\mu}_t^c(\mathcal{H}_t),\tilde{\Sigma}_t^c(\mathcal{H}_t)\big),~~~t=0,\dots,T-1
\end{equation}
under the risk--neutral probability measure $\tilde{\mathbb{P}}$, where $\tilde{\mu}_t^c(\mathcal{H}_t):=\big(\tilde{\mu}_{t+1|t}'(\mathcal{H}_t),\dots,\tilde{\mu}_{T|t}'(\mathcal{H}_t)\big)'$ is a conditional expectation and $\tilde{\Sigma}_t^c(\mathcal{H}_t):=\big(\tilde{\Sigma}_{i_1,i_2|t}(\mathcal{H}_t)\big)_{i_1,i_2=t+1}^T$ is a conditional covariance matrix of the random vector $\bar{x}_t^c$ and are calculated by equations \eqref{05025} and \eqref{05026}, respectively. 

\subsection{Forward Probability Measure}

According to \citeA{Geman95}, cleaver change of probability measure leads to a significant reduction in the computational burden of derivative pricing. The frequently used probability measure that reduces the computational burden is the forward probability measure and to price the zero--coupon bond, the European options, and the Margrabe exchange options we will apply it. To define the forward probability measure, we need to zero--coupon bond. It is the well--known fact that conditional on $\mathcal{H}_t$, price of zero--coupon bond paying face value 1 at time $t$ is $B_t(\mathcal{H}_t):=\frac{1}{D_t}\mathbb{\tilde{E}}\big[D_T\big|\mathcal{H}_t\big]$. The $t$--forward probability measure is defined by
\begin{equation}\label{05028}
\mathbb{\hat{P}}_t\big[A\big|\mathcal{H}_t\big]:=\frac{1}{D_tB_t(\mathcal{H}_t)}\int_AD_T\mathbb{\tilde{P}}\big[\omega|\mathcal{H}_t\big]~~~\text{for all}~A\in \mathcal{H}_T.
\end{equation}
Therefore, a negative exponent of $D_T/D_t$ in the zero--coupon bond formula is represented by 
\begin{equation}\label{05029}
\sum_{\beta=t+1}^T\tilde{r}_\beta=\tilde{r}_{t+1}+j_r'\Bigg[\sum_{\beta=t+1}^{T-1}J_{\beta|t}\Bigg]\bar{x}_t^c=\tilde{r}_{t+1}+\gamma_t'\bar{x}_t^c
\end{equation}
where $j_r:=(0,1)'$ is $(1\times \tilde{n})$ vector and it can be used to extract the log spot rate process $\tilde{r}_s$ from the random process $x_s$, $J_{\beta|t}:=[0:I_{\tilde{n}}:0]$ is $(\tilde{n}\times \tilde{n}(T-t))$ matrix, whose $(\beta-t)$--th block matrix equals $I_{\tilde{n}}$ and others are zero and it is used to extract the random vector $x_\beta$ from the random vector $\bar{x}_t^c$, and $\gamma_t':=j_r'\sum_{\beta=t+1}^{T-1}J_{\beta|t}$. Therefore, two times of negative exponent of the price at time $t$ of the zero--coupon bond $B_t(\mathcal{H}_t)$ is represented by
\begin{eqnarray}\label{05030}
&&2\sum_{s=t+1}^T\tilde{r}_s+\big(\bar{x}_t^c-\tilde{\mu}_t^c(\mathcal{H}_t)\big)'\big(\tilde{\Sigma}_t^c(\mathcal{H}_t)\big)^{-1}\big(\bar{x}_t^c-\tilde{\mu}_t^c(\mathcal{H}_t)\big)\nonumber\\
&&=\Big(\bar{x}_t^c-\tilde{\mu}_t^c(\mathcal{H}_t)+\tilde{\Sigma}_t^c(\mathcal{H}_t)\gamma_t\Big)'\big(\tilde{\Sigma}_t^c(\mathcal{H}_t)\big)^{-1}\Big(\bar{x}_t^c-\tilde{\mu}_t^c(\mathcal{H}_t)+\tilde{\Sigma}_t^c(\mathcal{H}_t)\gamma_t\Big)\\
&&+2\big(\tilde{r}_{t+1}+\gamma_t'\tilde{\mu}_t^c(\mathcal{H}_t)\big)-\gamma_t'\tilde{\Sigma}_t^c(\mathcal{H}_t)\gamma_t.\nonumber
\end{eqnarray}
As a result, for given $\mathcal{H}_t$, price at time $t$ of the zero--coupon $B_t(\mathcal{H}_t)$ is
\begin{equation}\label{05031}
B_t(\mathcal{H}_t)=\exp\bigg\{-\tilde{r}_{t+1}-\gamma_t'\tilde{\mu}_t^c(\mathcal{H}_t)+\frac{1}{2}\gamma_t'\tilde{\Sigma}_t^c(\mathcal{H}_t)\gamma_t\bigg\}.
\end{equation}
Consequently, conditional on the information $\mathcal{H}_t$, a joint distribution of the random vector $\bar{x}_t^c$ is given by
\begin{equation}\label{05032}
\bar{x}_t^c~|~\mathcal{H}_t\sim \mathcal{N}\big(\hat{\mu}_t^c(\mathcal{H}_t),\tilde{\Sigma}_t^c(\mathcal{H}_t)\big),~~~t=0,\dots,T-1
\end{equation}
under the $t$--forward probability measure $\hat{\mathbb{P}}_t$, where $\hat{\mu}_t^c(\mathcal{H}_t):=\tilde{\mu}_t^c(\mathcal{H}_t)-\tilde{\Sigma}_t^c(\mathcal{H}_t)\gamma_t$ and $\tilde{\Sigma}_t^c(\mathcal{H}_t)$ are conditional expectation and conditional covariance matrix, respectively, of the random vector $\bar{x}_t^c$. Also, as  $J_{s_1|t}\tilde{\Sigma}_t^c(\mathcal{H}_t) J_{s_2|t}'=\tilde{\Sigma}_{s_1,s_2|t}(\mathcal{H}_t)$, we have
\begin{equation}\label{05033}
J_{s|t}\tilde{\Sigma}_t^c(\mathcal{H}_t)\Bigg(\sum_{\beta=t+1}^{T-1}J_{\beta|t}'\Bigg)=\sum_{\beta=t+1}^{T-1}\tilde{\Sigma}_{s,\beta|t}(\mathcal{H}_t),
\end{equation}
where $\tilde{\Sigma}_{s,\beta|t}(\mathcal{H}_t)$ is calculated by equation \eqref{05026}. Therefore, $(s-t)$--th sub vector of the conditional expectation $\hat{\mu}_t^c(\mathcal{H}_t)$ is given by
\begin{equation}\label{05034}
\hat{\mu}_{s|t}(\mathcal{H}_t):=J_{s|t}\hat{\mu}_t^c(\mathcal{H}_t)=\tilde{\mu}_{s|t}(\mathcal{H}_t)-\sum_{\beta=t+1}^{T-1}\big(\tilde{\Sigma}_{s,\beta|t}(\mathcal{H}_t)\big)_{\tilde{n}},
\end{equation}
where for a generic matrix $O$, we denote its $j$--th column by $(O)_j$. Similarly, it is clear that price at time $t$ of the zero--coupon bond is given by
\begin{equation}\label{05037}
B_t(\mathcal{H}_t)=\exp\bigg\{-\tilde{r}_{t+1}-\sum_{\beta=t+1}^{T-1}\big(\tilde{\mu}_{\beta|t}(\mathcal{H}_t)\big)_{\tilde{n}}+\frac{1}{2}\sum_{\alpha=t+1}^{T-1}\sum_{\beta=t+1}^{T-1}\big(\tilde{\Sigma}_{\alpha,\beta|t}(\mathcal{H}_t)\big)_{\tilde{n},\tilde{n}}\bigg\}.
\end{equation}
where for a generic vector $o$, we denote its $j$--th element by $(o)_j$, and for a generic square matrix $O$, we denote its $(i,j)$--th element by $(O)_{i,j}$.

To price the European call and put options for asset value, we need a distribution of the log market value process at time $T$. For this reason, it follows from equation \eqref{05032} that the distribution of the log market value process at time $T$ is given by 
\begin{equation}\label{05038}
\tilde{V}_T~|~\mathcal{H}_t\sim \mathcal{N}\big(\hat{\mu}_{T|t}^{\tilde{V}}(\mathcal{H}_t),\tilde{\Sigma}_{T|t}^{\tilde{V}}(\mathcal{H}_t)\big),
\end{equation}
under the $t$--forward probability measure $\mathbb{\hat{P}}$, where $\hat{\mu}_{T|t}^{\tilde{V}}(\mathcal{H}_t):=J_V\hat{\mu}_{T|t}(\mathcal{H}_t)$ is a conditional expectation, which is calculated from equation \eqref{05034} and $\tilde{\Sigma}_{T|t}^{\tilde{V}}(\mathcal{H}_t):=J_V\tilde{\Sigma}_{T,T|t}^{\tilde{V}}(\mathcal{H}_t)J_V'$ is a conditional covariance matrix, which is calculated from equation \eqref{05026} of the log market value at time $T$ given the information $\mathcal{H}_t$ and $J_V:=[I_{2n}:0]$ is a ($2n\times \tilde{n}$) matrix, which is used to extract the log market value process $\tilde{V}_t$ from the process $x_t$.

\subsection{The European Call and Put Options}

Let us assume that the recovery rates, corresponding to the market values of assets of the companies are zero when they default. Then, because market values at time $T$ of the equities and liabilities are given by the following equations
$$V_T^e=\max(V_T^a-L,0)=(V_T^a-L)^+~~~\text{and}~~~L_T=\min(V_T^a,L)=L-(L-V_T^a)^+,$$
respectively, where $L$ is a nominal value vector of the liabilities at maturity $T$ of the companies and for a real vector $x$, $x^+$ denotes component--wise maximum of $x$ and zero. Therefore, a risk--neutral equity value at time $t$ of a public company equals the European call option on its asset and liabilities at time $t$ of the company is represented in terms of the European put option on its asset. This subsection is devoted to price the call and put options.

According to equation \eqref{05010} and \eqref{05038}, conditional on the information $\mathcal{H}_t$, its distribution is given by
\begin{equation}\label{05039}
\tilde{V}_T^a~|~\mathcal{H}_t\sim \mathcal{N}\big(\hat{\mu}_{T|t}^a(\mathcal{H}_t),\tilde{\Sigma}_{T|t}^{a}(\mathcal{H}_t)\big)
\end{equation}
under the $t$--forward probability measure $\hat{\mathbb{P}}_t$, where $\hat{\mu}_{T|t}^a(\mathcal{H}_t):=W_T^a\hat{\mu}_{T|t}^{\tilde{V}}(\mathcal{H}_t)+(G_T^a)^{-1}h_T^a$ and $\tilde{\Sigma}_{T|t}^a(\mathcal{H}_t):=W_T^a\tilde{\Sigma}_{T|t}^{\tilde{V}}(\mathcal{H}_t)(W_T^a)'$ are conditional mean and covariance matrix of the log asset value $\tilde{V}_T^a$, respectively, given the information $\mathcal{H}_t$. Therefore, due to equation \eqref{05039} and Lemma \ref{lem01}, see Technical Annex, price vectors at time $t$ of the Black--Scholes call and put options with strike price vector $L$ and maturity $T$ are given by
\begin{eqnarray}\label{05040}
C_{T|t}(\mathcal{H}_t)&=&\mathbb{\tilde{E}}\bigg[\frac{D_T}{D_t}\Big(V_{T}^a-L\Big)^+\bigg|\mathcal{H}_t\bigg]=B_t(\mathcal{H}_t)\mathbb{\hat{E}}\big[\big(V_T^a-L\big)^+\big|\mathcal{H}_t\big]\\
&=&B_t(\mathcal{H}_t)\bigg(\exp\bigg\{\hat{\mu}_{T|t}^a(\mathcal{H}_t)+\frac{1}{2}\mathcal{D}\big[\tilde{\Sigma}_{T|t}^a(\mathcal{H}_t)\big]\bigg\}\odot\Phi\big(d_{T|t}^1(\mathcal{H}_t)\big)-L\odot\Phi\big(d_{T|t}^2(\mathcal{H}_t)\big)\bigg)\nonumber
\end{eqnarray}
and
\begin{eqnarray}\label{05041}
P_{T|t}(\mathcal{H}_t)&=&\mathbb{\tilde{E}}\bigg[\frac{D_T}{D_t}\Big(L-V_T^a\Big)^+\bigg|\mathcal{H}_t\bigg]=B_t(\mathcal{H}_t)\mathbb{\hat{E}}\big[\big(L-V_T^a\big)^+\big|\mathcal{H}_t\big]\\
&=&B_t(\mathcal{H}_t)\bigg(L\odot\Phi\big(-d_{T|t}^2(\mathcal{H}_t)\big)-\exp\bigg\{\hat{\mu}_{T|t}^a(\mathcal{H}_t)+\frac{1}{2}\mathcal{D}\big[\tilde{\Sigma}_{T|t}^a(\mathcal{H}_t)\big]\bigg\}\odot\Phi\big(-d_{T|t}^1(\mathcal{H}_t)\big)\bigg),\nonumber
\end{eqnarray}
where $d_{T|t}^1(\mathcal{H}_t):=\Big(\hat{\mu}_{T|t}^a(\mathcal{H}_t)+\mathcal{D}\big[\tilde{\Sigma}_{T|t}^a(\mathcal{H}_t)\big]-\ln(L)\Big)\oslash\sqrt{\mathcal{D}\big[\tilde{\Sigma}_{T|t}^a(\mathcal{H}_t)\big]}$ and $d_{T|t}^2(\mathcal{H}_t):=d_{T|t}^1(\mathcal{H}_t)-\sqrt{\mathcal{D}\big[\tilde{\Sigma}_{T|t}^a(\mathcal{H}_t)\big]}$. 

Therefore, due to Lemma \ref{lem03} and the tower property of conditional expectation, price vectors at time $t$ ($t=0,\dots,T-1$) of the Black--Scholes call and put options on asset values with strike price vector $L$ and maturity $T$ are obtained as
\begin{equation}\label{05044}
C_{T|t}(\mathcal{F}_t)=\tilde{\mathbb{E}}\bigg[\frac{D_T}{D_t}\Big(V_{T}^a-L\Big)^+\bigg|\mathcal{F}_t\bigg]=\sum_{s}\int_{C_{\hat{s}},\Sigma_{\hat{s}}}C_{T|t}(\mathcal{H}_t)\tilde{f}(C_{\hat{s}},\Sigma_{\hat{s}},s|\mathcal{F}_t)dC_{\hat{s}} d\Sigma_{\hat{s}}
\end{equation}
and
\begin{equation}\label{05045}
P_{T|t}(\mathcal{F}_t)=\tilde{\mathbb{E}}\bigg[\frac{D_T}{D_t}\Big(L-V_{T}^a\Big)^+\bigg|\mathcal{F}_t\bigg]=\sum_{s}\int_{C_{\hat{s}},\Sigma_{\hat{s}}}P_{T|t}(\mathcal{H}_t)\tilde{f}(C_{\hat{s}},\Sigma_{\hat{s}},s|\mathcal{F}_t)dC_{\hat{s}} d\Sigma_{\hat{s}}
\end{equation}
respectively. 

As a result, according to formulas of the call and put options given in equations \eqref{05044} and \eqref{05045}, risk--neutral market values of the equities and liabilities at time $t$ of the companies are given by
\begin{equation}\label{05046}
\hat{V}_t^e=C_{T|t}(\mathcal{F}_t)~~~\text{and}~~~\hat{L}_t=LB_t(\mathcal{F}_t)-P_{T|t}(\mathcal{F}_t).
\end{equation}

\subsection{Default Probability}

Now, we move to default probabilities of the companies. In order to obtain the default probabilities of the companies, for given the information $\mathcal{H}_t$, we need a distribution of log asset value at time $T$ under the real probability measure $\mathbb{P}$. For this reason, let us write system \eqref{05009} in VAR(1) form
\begin{equation}\label{05047}
Q_{0,t}x_t=\nu_t+Q_{1,t}x_{t-1}+\mathsf{G}_t\xi_t
\end{equation} 
under the real probability measure $\mathbb{P}$, where $\nu_t:=(\nu_{V,t}',\nu_{r,t})'$ is intercept process of the VAR(1) process $x_t$ and
\begin{equation}\label{05048}
Q_{0,t}:=\begin{bmatrix}
I_{2n} & -G_t\delta \\
0 & 1
\end{bmatrix}~~~\text{and}~~~
Q_{1,t}:=\begin{bmatrix}
G_t & 0 \\
0 & 1
\end{bmatrix}
\end{equation}
are $(\tilde{n}\times \tilde{n})$ coefficient matrix. By repeating equation \eqref{05047}, one gets that for $i=t+1,\dots,T$,
\begin{equation}\label{05049}
x_i=\Pi_{t,i}x_t+\sum_{\beta=t+1}^i\Pi_{\beta,i}\nu_\beta+\sum_{\beta=t+1}^i\Pi_{\beta,i}\mathsf{G}_\beta\xi_\beta,
\end{equation}
where the coefficient matrices are for $\beta=t$,
\begin{equation}\label{05050}
\Pi_{\beta,i}:=\prod_{\alpha=\beta+1}^iQ_{0,\alpha}^{-1}Q_{1,\alpha}=\begin{bmatrix}
\displaystyle\prod_{\alpha=\beta+1}^iG_\alpha & \displaystyle \sum_{\alpha=\beta+1}^i\Bigg(\prod_{j=\alpha}^iG_j\Bigg)\delta\\
\displaystyle 0 & 1
\end{bmatrix}
\end{equation}
for $\beta=t+1,\dots,i-1$,
\begin{equation}\label{05050}
\Pi_{\beta,i}:=\Bigg(\prod_{\alpha=\beta+1}^iQ_{0,\alpha}^{-1}Q_{1,\alpha}\Bigg)Q_{0,\beta}^{-1}=\begin{bmatrix}
\displaystyle\prod_{\alpha=\beta+1}^iG_\alpha & \displaystyle \sum_{\alpha=\beta}^i\Bigg(\prod_{j=\alpha}^iG_j\Bigg)\delta\\
\displaystyle 0 & 1
\end{bmatrix},
\end{equation} 
and for $\beta=i$,
\begin{equation}\label{•}
\Pi_{\beta,i}:=Q_{0,\beta}^{-1}=\begin{bmatrix}
I_{2n} & G_\beta\delta \\
0 & 1
\end{bmatrix}.
\end{equation}
Thus, conditional on the information $\mathcal{H}_t$, for $i=t+1,\dots,T$, a expectation at time $i$ and a conditional covariance matrix at time $i_1$ and $i_2$ of the process $x_t$ is given by the following equations
\begin{equation}\label{05051}
\mu_{i|t}(\mathcal{H}_t):=\mathbb{E}\big[x_i\big|\mathcal{H}_t\big]=\Pi_{t,i} x_t+\sum_{\beta=t+1}^i\Pi_{\beta,i}\nu_\beta
\end{equation}
and
\begin{equation}\label{05052}
\Sigma_{i_1,i_2|t}(\mathcal{H}_t):=\text{Cov}\big[x_{i_1},x_{i_2}\big|\mathcal{H}_t\big]=\sum_{\beta=t+1}^{i_1\wedge i_2}\Pi_{\beta,i_1}\mathsf{G}_\beta\Sigma_{s_\beta}\mathsf{G}_\beta\Pi_{\beta,i_2}'.
\end{equation}

Therefore, it follows from equations \eqref{05051} and \eqref{05052} that conditional on the information $\mathcal{H}_t$, an expectation and covariance matrix of log market value process at time $T$ under the real probability measure $\mathbb{P}$ are given by the following equations
\begin{equation}\label{05053}
\mu_{T|t}^{\tilde{V}}(\mathcal{H}_t):=\mathbb{E}\big[\tilde{V}_T\big|\mathcal{H}_t\big]=J_V\mu_{T|t}(\mathcal{H}_t)
\end{equation}
and
\begin{eqnarray}\label{05054}
\Sigma_{T|t}^{\tilde{V}}(\mathcal{H}_t)&:=&\text{Var}\big[\tilde{V}_T\big|\mathcal{H}_t\big]=J_V\Sigma_{T,T|t}(\mathcal{H}_t)J_V'.
\end{eqnarray}
Consequently, due to equation \eqref{05010}, conditional on $\mathcal{H}_t$, a distribution of the log asset value process at time $T$ is given by
\begin{equation}\label{05055}
\tilde{V}_T^a~|~\mathcal{H}_t\sim \mathcal{N}\Big(\mu_{T|t}^a(\mathcal{H}_t),\Sigma_{T|t}^a(\mathcal{H}_t)\Big)
\end{equation}
under the real probability measure $\mathbb{P}$, where $\mu_{T|t}^a(\mathcal{H}_t):=W_T^a\mu_{T|t}^{\tilde{V}}(\mathcal{H}_t)+(G_T^a)^{-1}h_T^a$ and $\Sigma_{T|t}^a(\mathcal{H}_t):=W_T^a\Sigma_{T|t}^{\tilde{V}}(\mathcal{H}_t)(W_T^a)'$ are conditional mean and covariance matrix of the log asset value $\tilde{V}_T^a$, respectively, given the information $\mathcal{H}_t$.  According to the structural model of default risk, if the asset value of a company falls below the default threshold, representing liabilities, then default occurs. Therefore, due to equation \eqref{05055}, conditional on the information $\mathcal{H}_t$, a joint default probability at time $t$ of the companies is given by the following equation
\begin{equation}\label{05056}
\mathbb{P}\big[V_T^a\leq \bar{L}|\mathcal{H}_t\big]=\mathbb{P}\big[\tilde{V}_T^a\leq \ln(\bar{L})|\mathcal{H}_t\big]=\Phi_n\left(\big(\Sigma_{T|t}^a(\mathcal{H}_t)\big)^{-1}\big(\ln(\bar{L})-\mu_{T|t}^a(\mathcal{H}_t)\big)\right),
\end{equation}
where $\bar{L}$ is the default threshold vector at maturity $T$ and for a random vector $Z\sim\mathcal{N}(0,I_n)$, $\Phi_n(z):=\mathbb{P}(Z\leq z)$ is a joint distribution function of the random vector $Z$. As a result, by the tower property of conditional expectation formula and Lemma \ref{lem04}, we get that
\begin{equation}\label{05057}
\mathbb{P}\big[V_T^a\leq \bar{L}|\mathcal{F}_t\big]=\sum_{s}\int_{C_{\hat{s}},\Sigma_{\hat{s}}}\mathbb{P}\big[V_T^a\leq \bar{L}|\mathcal{H}_t\big]f(C_{\hat{s}},\Sigma_{\hat{s}},s|\mathcal{F}_t)dC_{\hat{s}} d\Sigma_{\hat{s}}.
\end{equation}

\section{Parameter Estimation}

To estimate parameters of the required rate of return $\tilde{k}_t$, \citeA{Battulga23b} used the maximum likelihood method and Kalman filtering. For Bayesian method, which removes duplication in regime vector, we refer to \citeA{Battulga24g}. In this section, we assume that coefficient matrices $C_1,\dots,C_N$, covariance matrices $\Sigma_1,\dots,\Sigma_N$, and transition probability matrix $\mathsf{P}$ are deterministic. Here we apply the EM algorithm to estimate parameters of the model. If we combine the equations \eqref{05006} and \eqref{05008}, then we have that
\begin{equation}
B_0y_t=C_{s_t}\psi_t+B_1y_{t-1}+\xi_t,
\end{equation}
where $y_t:=(\tilde{k}_t',\tilde{r}_t)'$ is an $(\tilde{n}\times 1)$ vector of endogenous variables, $C_{s_t}$ is the $(\tilde{n}\times l)$ matrix, which depends on the regime $s_t$, and the $(\tilde{n}\times \tilde{n})$ matrices are given by 
\begin{equation}\label{•}
B_0:=\begin{bmatrix}
I_{2n} & -\delta\\
0 & 1
\end{bmatrix}~~~\text{and}~~~
B_1:=\begin{bmatrix}
0 & 0\\
0 & 1
\end{bmatrix}.
\end{equation}
For $t=0,\dots,T$, let $\mathcal{Y}_t$ be the available data at time $t$, which is used to estimate parameters of the model, that is, $\mathcal{Y}_t:=\sigma(y_0,y_1,\dots,y_t).$ Then, it is clear that the log--likelihood function of our model is given by the following equation
\begin{equation}\label{02018}
\mathcal{L}(\theta)=\sum_{t=1}^T\ln\big(f(y_t|\mathcal{Y}_{t-1};\theta)\big)
\end{equation}
where $\theta:=\big(\text{vec}(C_1)',\dots,\text{vec}(C_N)',\text{vec}(\Sigma_1)',\dots,\text{vec}(\Sigma_N)',\text{vec}(\mathsf{P})'\big)'$ is a vector, which consists of all population parameters of the model and $f(y_t|\mathcal{Y}_{t-1};\theta)$ is a conditional density function of the random vector $y_t$ given the information $\mathcal{Y}_{t-1}$.  The log--likelihood function is used to obtain the maximum likelihood estimator of the parameter vector $\theta$. Note that the log--likelihood function depends on all observations, which are collected in $\mathcal{Y}_T$, but does not depend on regime--switching process $s_t$, whose values are unobserved. If we assume that the regime--switching process in regime $j$ at time $t$, then because conditional on the information $\mathcal{Y}_{t-1}$, $\xi_t$ follows a multivariate normal distribution with mean zero and covariance matrix $\Sigma_j$, the conditional density function of the random vector $y_t$ is given by the following equation
\begin{eqnarray}\label{02019}
\eta_{t,j}&:=&f(y_t|s_t=j,\mathcal{Y}_{t-1};\alpha)\\
&=&\frac{1}{(2\pi)^{\tilde{n}/2}|\Sigma_j|^{1/2}}\exp\bigg\{-\frac{1}{2}\Big(B_0y_t-C_j\psi_t-B_1y_{t-1}\Big)'\Sigma_j^{-1}\Big(B_0y_t-C_j\psi_t-B_1y_{t-1}\Big)\bigg\}\nonumber
\end{eqnarray}
for $t=1,\dots,T$ and $j=1,\dots,N$, where $\alpha:=\big(\text{vec}(C_1)',\dots,\text{vec}(C_N)',\text{vec}(\Sigma_1)',\dots,\text{vec}(\Sigma_N)'\big)'$ is a parameter vector, which differs from the vector of all parameters $\theta$ by the transition probability matrix $\mathsf{P}$. For all $t=1,\dots,T$, we collect the conditional density functions of the price at time $t$ into an $(N\times 1)$ vector $\eta_t$, that is, $\eta_t:=(\eta_{t,1},\dots,\eta_{t,N})'$. 

Let us denote a probabilistic inference about the value of the regime--switching process $s_t$ is equal to $j$, based on the information $\mathcal{Y}_t$ and the parameter vector $\theta$ by $\mathbb{P}(s_t=j|\mathcal{Y}_t,\theta)$. Collect these conditional probabilities $\mathbb{P}(s_t=j|\mathcal{Y}_t,\theta)$ for $j=1,\dots,N$ into an $(N\times 1)$ vector $z_{t|t}$, that is, $z_{t|t}:=\big(\mathbb{P}(s_t=1|\mathcal{Y}_t;\theta),\dots,\mathbb{P}(s_t=N|\mathcal{Y}_t;\theta)\big)'$. Also, we need a probabilistic forecast about the value of the regime--switching process at time $t+1$ is equal to $j$ conditional on data up to and including time $t$. Collect these forecasts into an $(N\times 1)$ vector $z_{t+1|t}$, that is, $z_{t+1|t}:=\big(\mathbb{P}(s_{t+1}=1|\mathcal{Y}_t;\theta),\dots,\mathbb{P}(s_{t+1}=N|\mathcal{Y}_t;\theta)\big)'$.  

The probabilistic inference and forecast for each time $t=1,\dots,T$ can be found by iterating on the following pair of equations: 
\begin{equation}\label{02021}
z_{t|t}=\frac{(z_{t|t-1}\odot\eta_t)}{i_N'(z_{t|t-1}\odot\eta_t)}~~~\text{and}~~~z_{t+1|t}=\hat{\mathsf{P}}'z_{t|t},~~~t=1,\dots,T,
\end{equation}
see book of \citeA{Hamilton94}, where $\eta_t$ is the $(N\times 1)$ vector, whose $j$-th element is given by equation \eqref{02019}, $\hat{\mathsf{P}}$ is the $(N\times N)$ transition probability matrix, which is defined by omitting the first row of the matrix $\mathsf{P}$, and $i_N$ is an $(N\times 1)$ vector, whose elements equal 1. Given a starting value $z_{1|0}$ and an assumed value for the population parameter vector $\theta$, one can iterate on \eqref{02021} for $t=1,\dots,T$ to calculate the values of $z_{t|t}$ and $z_{t+1|t}$. 

To obtain MLE of the population parameters, in addition to the inferences and forecasts we need a smoothed inference about the regime--switching process at time $t$ is equal to $j$ based on full information $\mathcal{Y}_T$. Collect these smoothed inferences into an $(N\times 1)$ vector $z_{t|T}$, that is, $z_{t|T}:=\big(\mathbb{P}(s_t=1|\mathcal{Y}_T;\theta),\dots,\mathbb{P}(s_t=N|\mathcal{Y}_T;\theta)\big)'$. The smoothed inferences can be obtained by using the \citeA{Battulga24g}'s exact smoothing algorithm:
\begin{equation}\label{08110}
z_{T-1|T}=\frac{1}{i_N'(z_{T|T-1}\odot\eta_t)}\big(\hat{\mathsf{P}}\mathsf{H}_Ti_N\big)\odot z_{T-1|T-1}
\end{equation}
and for $t=T-2,\dots,1$,
\begin{equation}\label{08180}
z_{t|T}=\frac{1}{i_N'(z_{t+1|t}\odot\eta_{t+1})}\Big(\hat{\mathsf{P}}\mathsf{H}_{t+1}\big(z_{t+1|T}\oslash z_{t+1|t+1}\big)\Big)\odot z_{t|t},
\end{equation}
where $\oslash$ is an element--wise division of two vectors and $\mathsf{H}_{t+1}:=\mathrm{diag}\{\eta_{t+1,1},\dots,\eta_{t+1,N}\}$ is an $(N\times N)$ diagonal matrix. For $t=2,\dots,T$, joint probability of the regimes $s_{t-1}$ and $s_t$ is
\begin{equation}\label{08165}
\mathbb{P}(s_{t-1}=i,s_t=j|\mathcal{F}_t;\theta)=\frac{(z_{t|T})_j\eta_{t,j}p_{s_{t-1}s_t}(z_{t-1|t-1})_i}{(z_{t|t})_ji_N'(z_{t|t-1}\odot\eta_t)},
\end{equation}
where for a generic vector $o$, $(o)_j$ denotes $j$--th element of the vector $o$.

The EM algorithm is an iterative method to obtain (local) maximum likelihood estimate of parameters of distribution functions, which depend on unobserved (latent) variables. The EM algorithm alternates an expectation (E) step and a maximization (M) step. In E--Step, we consider that conditional on the full information $\mathcal{Y}_T$ and parameter at iteration $k$, $\theta^{[k]}$, expectation of augmented log--likelihood of the data $\mathcal{Y}_T$ and unobserved (latent) transition probability matrix $\mathsf{P}$. The E--Step defines a objective function $\mathcal{L}$,
namely,
\begin{eqnarray}
\mathcal{L}&=&\mathbb{E}\bigg[-\frac{T\tilde{n}}{2}\ln(2\pi)-\frac{1}{2}\sum_{t=1}^T\sum_{j=1}^N\ln(\Sigma_j)1_{\{s_t=j\}}\nonumber\\
&-&\frac{1}{2}\sum_{t=1}^T\sum_{j=1}^N\Big(B_0y_t-C_j\psi_t-B_1y_{t-1}\Big)'\Sigma_j^{-1}\Big(B_0y_t-C_j\psi_t-B_1y_{t-1}\Big)1_{\{s_t=j\}}\\
&+&\sum_{j=1}^Np_{0j}1_{\{s_1=j\}}+\sum_{t=2}^T\sum_{i=1}^N\sum_{j=1}^N\ln(p_{ij})1_{\{s_{t-1}=i,s_t=j\}}-\sum_{i=0}^N\mu_i\bigg(\sum_{j=1}^Np_{ij}-1\bigg)\bigg|\mathcal{Y}_T;\theta^{[k]}\bigg]\nonumber
\end{eqnarray}
In M--Step, to obtain parameter estimate of next iteration $\theta^{[k+1]}$, one maximizes the objective function with respect to the parameter $\theta$. First, let us consider partial derivative from the objective function with respect to the parameter $C_j$ for $j=1,\dots,N$. Let $c_j$ is a vectorization of the matrix $C_j$, i.e., $c_j=\text{vec}(C_j)$. Since $C_j\psi_t=(\psi_t'\otimes I_{2n+1})c_j$, we have that
\begin{equation}
\frac{\partial \mathcal{L}}{\partial c_j'}=\sum_{t=1}^T\Big(B_0y_t-\big(\psi_t'\otimes I_{2n+1}\big)c_j-B_1 y_{t-1}\Big)'\Sigma_j^{-1}\big(\psi_t'\otimes I_{2n+1}\big)\big(z_{t|T}^{[k]}\big)_j.
\end{equation}
Consequently, an estimator at iteration $(k+1)$ of the parameter $c_j$ is given by
\begin{eqnarray}\label{•}
c_j^{[k+1]}&=&\bigg(\sum_{t=1}^T\big(\psi_t\otimes I_{2n+1}\big)\Sigma_j^{-1}\big(\psi_t\otimes I_{2n+1}\big)\big(z_{t|T}^{[k]}\big)_j\bigg)^{-1}\nonumber\\
&\times&\sum_{t=1}^T\big(\psi_t\otimes I_{2n+1}\big)\Sigma_j^{-1}\big(B_0y_t-B_1y_{t-1}\big)\big(z_{t|T}^{[k]}\big)_j.
\end{eqnarray}
As a result, an estimator at iteration $(k+1)$ of the parameter $C_j$ is given by
\begin{equation}
C_j^{[k+1]}=\big(B_0\bar{y}_j^{[k]}-B_1\bar{y}_{j,-1}^{[k]}\big)\big(\bar{\psi}_j^{[k]}\big)'\big(\bar{\psi}_j^{[k]}\big(\bar{\psi}_j^{[k]}\big)'\big)^{-1},
\end{equation}
where $\bar{y}_j^{[k]}:=\Big[y_1\sqrt{\big(z_{1|T}^{[k]}\big)_j}:\dots:y_T\sqrt{\big(z_{T|T}^{[k]}\big)_j}\Big]$ is a $(\tilde{n}\times T)$ matrix, $\bar{y}_{j,-1}^{[k]}:=\Big[y_0\sqrt{\big(z_{1|T}^{[k]}\big)_j}:\dots:y_{T-1}\sqrt{\big(z_{T|T}^{[k]}\big)_j}\Big]$ is a $(\tilde{n}\times T)$ matrix, and $\bar{\psi}_j^{[k]}:=\Big[\psi_1\sqrt{\big(z_{1|T}^{[k]}\big)_j}:\dots:\psi_T\sqrt{\big(z_{T|T}^{[k]}\big)_j}\Big]$ is an $(l\times T)$ matrix. Second, a partial derivative from the objective function with respect to the parameter $\Sigma_j$ for $j=1,\dots,N$ is given by
\begin{eqnarray}
\frac{\partial \mathcal{L}}{\partial \Sigma_j}&=&-\frac{1}{2}\Sigma_j^{-1}\sum_{t=1}^T\big(z_{t|T}^{[k]}\big)_j\nonumber\\
&+&\frac{1}{2}\sum_{t=1}^T\Sigma_j^{-1}\big(y_t-C_j\psi_t-D y_{t-1}\big)\big(y_t-C_j\psi_t-D y_{t-1}\big)'\Sigma_j^{-1}\big(z_{t|T}^{[k]}\big)_j.
\end{eqnarray}
Consequently, an estimator at iteration $(k+1)$ of the parameter $\Sigma_j$ is given by
\begin{equation}
\Sigma_j^{[k+1]}=\frac{1}{\sum_{t=1}^T\big(z_{t|T}^{[k]}\big)_j}\sum_{t=1}^T\big(B_0y_t-C_j^{[k+1]}\psi_t-B_1 y_{t-1}\big)\big(B_0y_t-C_j^{[k+1]}\psi_t-B_1 y_{t-1}\big)'\big(z_{t|T}^{[k]}\big)_j.
\end{equation}
Third, a partial derivative from the objective function with respect to the parameter $p_{ij}$ for $i,j=1,\dots,N$ is given by
\begin{eqnarray}
\frac{\partial \mathcal{L}}{\partial p_{ij}}=\frac{1}{p_{ij}}\sum_{t=2}^T\mathbb{P}\big(s_{t-1}=i,s_t=j|\mathcal{F}_T;\theta^{[k]}\big)-\mu_i.
\end{eqnarray}
Consequently, an estimator at iteration $(k+1)$ of the parameter $p_{ij}$ is given by
\begin{equation}\label{•}
p_{ij}^{[k+1]}=\frac{1}{\sum_{t=2}^T\big(z_{t|T}^{[k]}\big)_i}\sum_{t=2}^T\mathbb{P}\big(s_{t-1}=i,s_t=j|\mathcal{F}_T;\theta^{[k]}\big)
\end{equation}
where the joint probability $\mathbb{P}\big(s_{t-1}=i,s_t=j|\mathcal{F}_T;\theta^{[k]}\big)$ is calculated by equation \eqref{08165}. Fourth, a partial derivative from the objective function with respect to the parameter $p_{0j}$ for $j=1,\dots,N$ is given by
\begin{eqnarray}
\frac{\partial \mathcal{L}}{\partial p_{0j}}=\frac{1}{p_{0j}}\mathbb{P}\big(s_1=j|\mathcal{F}_T;\theta^{[k]}\big)-\mu_0.
\end{eqnarray}
Consequently, an estimator at iteration $(k+1)$ of the parameter $p_{0j}$ is given by
\begin{equation}
p_{0j}^{[k+1]}=\big(z_{1|T}^{[k]}\big)_j.
\end{equation}
Alternating between these steps, the EM algorithm produces improved parameter estimates at each step (in the sense that the value of the original log--likelihood is continually increased) and it converges to the maximum likelihood (ML) estimates of the parameters.

To use the suggested model, we need to calculate the mean log dividend--to--price ratio $\mu_t$ and the mean log liability value--to--equity value ratio $\mu_t^a$ applying the parameter estimation. According to equations \eqref{05005} and \eqref{05006}, we have that
\begin{equation}\label{06085}
\tilde{p}_t-\tilde{V}_t-h_t=G_t\big(\tilde{p}_{t-1}-\tilde{V}_{t-1}+\tilde{p}_t-\tilde{p}_{t-1}-C_{k,s_t}\psi_t-\delta\tilde{r}_t\big)-G_tu_t.
\end{equation}
By taking expectation with respect to the real probability measure $\mathbb{P}$, one finds that
\begin{equation}\label{06086}
\mu_t-h_t=G_t\big(\mu_{t-1}+\tilde{p}_t-\tilde{p}_{t-1}-\mathbb{E}[C_{k,s_t}|\mathcal{F}_0]\psi_t-\delta\mathbb{E}[\tilde{r}_t|\mathcal{F}_0]\big),
\end{equation}
where according to equation \eqref{05008}, the expectations $\mathbb{E}[C_{k,s_t}|\mathcal{F}_0]$ and $\mathbb{E}[\tilde{r}_t|\mathcal{F}_0]$ equal 
\begin{equation}\label{ad006}
\mathbb{E}[C_{k,s_t}|\mathcal{F}_0]=\sum_{j=1}^NC_{k,j}\mathbb{P}[s_t=j|\mathcal{F}_0]=\sum_{j=1}^NC_{k,j}\big(p_0\hat{\mathsf{P}}^t\big)_j
\end{equation}
and
\begin{equation}\label{•}
\mathbb{E}[\tilde{r}_t|\mathcal{F}_0]=\tilde{r}_1+\sum_{i=2}^t\mathbb{E}[c_{r,s_i}'|\mathcal{F}_0]\psi_i=\tilde{r}_1+\sum_{i=2}^t\bigg(\sum_{j=1}^Nc_{r,j}'\big(p_0\hat{\mathsf{P}}^t\big)_j\bigg)\psi_i,
\end{equation}
respectively. On the other hand, the definition of the linearization parameter $h_t$ implies that
\begin{equation}\label{06087}
\mu_t-h_t=G_t\big(\mu_t-\ln(g_t)\big).
\end{equation}
Therefore, for successive values of the parameter $\mu_t$, it holds
\begin{equation}\label{06088}
\mu_t=\mu_{t-1}+\ln(g_t)+\tilde{p}_t-\tilde{p}_{t-1}-\mathbb{E}[C_{k,s_t}|\mathcal{F}_0]\psi_t-\delta\mathbb{E}[\tilde{r}_t|\mathcal{F}_0],
\end{equation}
where the above recurrence equation's initial value is $\mu_0=\tilde{p}_0-\tilde{V}_0$. One may be solve the above nonlinear system of equations by numerical methods for the parameter $\mu_t$. Here we consider Newton's iteration method to obtain solution of the system of equations. If we substitute equation $g_t=i_n+\exp\{\mu_t\}$ into the above system of equations, then we have that
\begin{equation}\label{•}
M_t(\mu):=\mu-\ln\big(i_n+\exp(\mu)\big)-\mu_{t-1}-\tilde{p}_t+\tilde{p}_{t-1}+\mathbb{E}[C_{k,s_t}|\mathcal{F}_0]\psi_t+\delta\mathbb{E}[\tilde{r}_t|\mathcal{F}_0]=0.
\end{equation}
Since an inverse matrix of Jacobian of the function $M_t(\mu)$ is $J(\mu)^{-1}:=\text{diag}\{i_n+\exp(\mu)\}$, Newton's iteration is given by
\begin{equation}\label{•}
\mu_{j+1,t}=\mu_{j,t}-J(\mu_{j,t})^{-1}M_t(\mu_{j,t}),
\end{equation}
where $\mu_{0,t}$ is an initial guess value of the mean log dividend--to--price ratio $\mu_t$. 

\section{Conclusion}

In this paper, we developed the Merton's structural model for public companies under an assumption that liabilities of the companies are observed. By modeling the market values of equities, liabilities and assets of companies using the \citeauthor{Campbell88}'s \citeyear{Campbell88} approximation method, we obtain formulas for risk--neutral equity and liability values and default probabilities of the companies. Finally, we study ML estimators of suggested model's parameters. It is worth mentioning that following the ideas in \citeA{Battulga24e} one can develop option pricing formulas with default risk and portfolio selection theory with default risk for public companies.

\section{Technical Annex}

Here we give the Lemmas, which are used in the paper.

\begin{lem}\label{lem01}
Let $X\sim \mathcal{N}(\mu,\sigma^2)$. Then for all $K>0$,
\begin{equation*}\label{05085}
\mathbb{E}\big[\big(e^X-K\big)^+\big]=\exp\bigg\{\mu+\frac{\sigma^2}{2}\bigg\}\Phi(d_1)-K\Phi(d_2)
\end{equation*}
and
\begin{equation*}\label{05086}
\mathbb{E}\big[\big(K-e^X\big)^+\big]=K\Phi(-d_2)-\exp\bigg\{\mu+\frac{\sigma^2}{2}\bigg\}\Phi(-d_1),
\end{equation*}
where $d_1:=\big(\mu+\sigma^2-\ln(K)\big)/\sigma$, $d_2:=d_1-\sigma$, and $\Phi(x)=\int_{-\infty}^x\frac{1}{\sqrt{2\pi}}e^{-u^2/2}du$ is the cumulative standard normal distribution function.
\end{lem}
\begin{proof}
See, e.g., \citeA{Battulga24a} and \citeA{Battulga24e}.
\end{proof}

Let us denote conditional on a generic $\sigma$-field $\mathcal{O}$, a joint density functions of a generic random vector $X$ by $f(X|\mathcal{O})$ and $\tilde{f}(X|\mathcal{O})$ under $\mathbb{P}$ and $\mathbb{\tilde{P}}$, respectively, and let $\mathcal{J}_t:=\sigma(\bar{C}_t)\vee \sigma(\bar{\Gamma}_t)\vee \sigma(\bar{s}_t)\vee \mathcal{F}_0$. Then, the following Lemmas hold.

\begin{lem}\label{lem03}
Conditional on $\mathcal{F}_t$, a joint density of $\big(\Pi_{\hat{s}},\Sigma_{\hat{s}},s,\mathsf{P}\big)$ is given by
\begin{equation}\label{07042}
\tilde{f}\big(C_{\hat{s}},\Sigma_{\hat{s}},s,\mathsf{P}|\mathcal{F}_t\big)=\frac{\tilde{f}(\bar{y}_t|C_{\alpha},\Sigma_{\alpha},\bar{s}_t,\mathcal{F}_0)f(C_{\hat{s}},\Sigma_{\hat{s}}|\hat{s},\mathcal{F}_0)f(s,\mathsf{P}|\mathcal{F}_0)}{\displaystyle \sum_{\bar{s}_t}\bigg(\int_{C_{\alpha},\Sigma_{\alpha}}\tilde{f}(\bar{y}_t|C_{\alpha},\Sigma_{\alpha},\bar{s}_t,\mathcal{F}_0)f(C_{\alpha},\Sigma_{\alpha}|\alpha,\mathcal{F}_0)dC_\alpha d\Sigma_\alpha\bigg)f(\bar{s}_t|\mathcal{F}_0)}
\end{equation}
for $t=1,\dots,T$, where for $t=1,\dots,T$,
\begin{equation}\label{07043}
\tilde{f}(\bar{y}_t|C_{\alpha},\Sigma_{\alpha},\bar{s}_t,\mathcal{F}_0)=\frac{1}{(2\pi)^{nt/2}|\Sigma_{11}|^{1/2}}\exp\Big\{-\frac{1}{2}\big(\bar{y}_t-\tilde{\mu}_1\big)'\tilde{\Sigma}_{11}^{-1}\big(\bar{y}_t-\tilde{\mu}_1\big)\Big\}
\end{equation}
with $\tilde{\mu}_1:=\big(\tilde{\mu}_{1|0}'(\mathcal{H}_0),\dots,\tilde{\mu}_{t|0}'(\mathcal{H}_0)\big)'$ and $\tilde{\Sigma}_{11}:=\big(\tilde{\Sigma}_{i_1,i_2|0}(\mathcal{H}_0)\big)_{i_1,i_2=1}^t$. In particular, we have that
\begin{equation}\label{ad001}
\tilde{f}\big(C_{\hat{s}},\Sigma_{\hat{s}},s|\mathcal{F}_t\big)=\frac{\tilde{f}(\bar{y}_t|C_{\alpha},\Sigma_{\alpha},\bar{s}_t,\mathcal{F}_0)f(C_{\hat{s}},\Sigma_{\hat{s}}|\hat{s},\mathcal{F}_0)f(s|\mathcal{F}_0)}{\displaystyle \sum_{\bar{s}_t}\bigg(\int_{C_{\alpha},\Sigma_{\alpha}}\tilde{f}(\bar{y}_t|C_{\alpha},\Sigma_{\alpha},\bar{s}_t,\mathcal{F}_0)f(C_{\alpha},\Sigma_{\alpha}|\alpha,\mathcal{F}_0)dC_\alpha d\Sigma_\alpha\bigg)f(\bar{s}_t|\mathcal{F}_0)}
\end{equation}
for $t=1,\dots,T$.
\end{lem}

\begin{proof}
See, \citeA{Battulga24a}.
\end{proof}

\begin{lem}\label{lem04}
Let
\begin{equation}\label{07043}
f(\bar{y}_t|C_{\alpha},\Sigma_{\alpha},\bar{s}_t,\mathcal{F}_0)=\frac{1}{(2\pi)^{nt/2}|\Sigma_{11}|^{1/2}}\exp\Big\{-\frac{1}{2}\big(\bar{y}_t-\mu_1\big)'\Sigma_{11}^{-1}\big(\bar{y}_t-\mu_1\big)\Big\},
\end{equation}
where $\mu_1:=\big(\mu_{1|0}'(\mathcal{H}_0),\dots,\mu_{t|0}'(\mathcal{H}_0)\big)'$ and $\Sigma_{11}:=\big(\Sigma_{i_1,i_2|0}(\mathcal{H}_0)\big)_{i_1,i_2=1}^t$. Then, we have that
\begin{equation}\label{ad001}
f\big(C_{\hat{s}},\Sigma_{\hat{s}},s|\mathcal{F}_t\big)=\frac{f(\bar{y}_t|C_{\alpha},\Sigma_{\alpha},\bar{s}_t,\mathcal{F}_0)f(C_{\hat{s}},\Sigma_{\hat{s}}|\hat{s},\mathcal{F}_0)f(s|\mathcal{F}_0)}{\displaystyle \sum_{\bar{s}_t}\bigg(\int_{C_{\alpha},\Sigma_{\alpha}}f(\bar{y}_t|C_{\alpha},\Sigma_{\alpha},\bar{s}_t,\mathcal{F}_0)f(C_{\alpha},\Sigma_{\alpha}|\alpha,\mathcal{F}_0)dC_\alpha d\Sigma_\alpha\bigg)f(\bar{s}_t|\mathcal{F}_0)}
\end{equation}
for $t=1,\dots,T$.
\end{lem}

\begin{proof}
By following \citeA{Battulga24a}, one can prove the Lemma 3.
\end{proof}

\bibliographystyle{apacite}
\bibliography{References}

\end{document}